\documentclass[12pt]{article}
\usepackage{fullpage}
\usepackage{authblk}

\usepackage{times}  
\usepackage{helvet} 
\usepackage{courier}  
\usepackage[hyphens]{url}  
\usepackage{graphicx} 
\usepackage{graphicx}  
\usepackage{amsthm}
\usepackage{amssymb}
\usepackage{mathrsfs}
\usepackage{amsmath}

\newtheorem{theorem}{Theorem}
\newtheorem{corollary}{Corollary}
\newtheorem{lemma}{Lemma}

\DeclareMathOperator*{\argmax}{argmax}

\usepackage{amsmath,amssymb,mathrsfs,mathtools,amsfonts,amsthm}
\usepackage[square,sort,comma,numbers]{natbib}
\allowdisplaybreaks

\usepackage{comment,color}

\usepackage{algpseudocode}

\begin{document}


\title{Bounded Incentives in Manipulating the Probabilistic Serial Rule}


\author[1]{Zihe Wang}
\author[2]{Zhide Wei}
\author[3]{Jie Zhang}
\affil[1]{Shanghai University of Finance and Economics, China\\ 

wang.zihe@mail.shufe.edu.cn }
\affil[2]{Peking University, China\\ 

zhidewei@pku.edu.cn}
\affil[3]{University of Southampton, U.K.\\ 

jie.zhang@soton.ac.uk}
\date{}


\maketitle
\begin{abstract}
The Probabilistic Serial mechanism is well-known for its desirable fairness and efficiency properties. It is one of the most prominent protocols for the random assignment problem. 
However, Probabilistic Serial is not incentive-compatible, thereby these desirable properties only hold for the agents' declared preferences, rather than their genuine preferences. 
A substantial utility gain through strategic behaviors would trigger self-interested agents to manipulate the mechanism and would subvert the very foundation of adopting the mechanism in practice.
In this paper, we characterize the extent to which an individual agent can increase its utility by strategic manipulation.
We show that the incentive ratio of the mechanism is $\frac{3}{2}$. That is, no agent can misreport its preferences such that its utility becomes more than 1.5 times of what it is when reports truthfully.
This ratio is a worst-case guarantee by allowing an agent to have complete information about other agents' reports and to figure out the best response strategy even if it is computationally intractable in general.
To complement this worst-case study, we further evaluate an agent's utility gain on average by experiments. The experiments show that an agent' incentive in manipulating the rule is very limited. 
These results shed some light on the robustness of Probabilistic Serial against strategic manipulation, which is one step further than knowing that it is not incentive-compatible.
 \end{abstract}

\section{Introduction}
Resource allocation is a fundamental and widely applicable area within AI and computer science.
The \emph{random assignment problem} is central in the resource allocation area, which has wide applications in allocating workers to shifts, houses to people, dormitories to students, ramp and hangars to airlines, and chores to cleaning staff.
In the problem, there are a set of agents and a set of items. The agents participate in a mechanism by reporting their private preferences over the items. The mechanism then assigns items to agents, according to a pre-defined allocation rule.

The random assignment problem, which is also known as the \emph{one-sided matching problem}, was introduced in \cite{HZ:79} and has been studied extensively ever since. Over the years, several mechanisms have been investigated, including Probabilistic Serial \cite{DBLP:journals/jet/BogomolnaiaM01,katta2006solution,BCK:11,DBLP:journals/jet/BogomolnaiaH12,BCKM2013,AS:13}, Random Priority \cite{RePEc:ecm:emetrp:v:66:y:1998:i:3:p:689-702,Sv99,aziz2013computational,DBLP:conf/atal/ChristodoulouFF16}, and Competitive Equilibrium from Equal Incomes (CEEI) \cite{BudishCEEI2011}. In the indivisible goods setting, the Top Trading Cycles (TTC) method is well-studied and generalized to investigate various problems. In particular, \citeauthor{RePEc:ecm:emetrp:v:66:y:1998:i:3:p:689-702} (1998) proposed an adaptation of the TTC method and established an equivalence between the adapted mechanism and Random Priority. \citeauthor{Kesten2009} (2009) proposed several extensions of these popular mechanisms and presented an equivalence result between those mechanisms in terms of economic efficiency.

There are various desired properties one would like these mechanisms to fulfill, especially when considering to deploy them in practice. In terms of \emph{efficiency}, an allocation \emph{Pareto improves} another if each agent has at least the same utility in the first, and there is at least one agent where the utility is greater than that of the second. An allocation is \emph{Pareto efficient} if there is no allocation which Pareto improves it. In terms of \emph{fairness}, an allocation is \emph{envy-free} if no agent prefers another agent’s allocation to its own. An allocation is \emph{proportional} if each of the $n$ agents receives at least $1/n$ of the resources by its own subjective valuation. Also, a mechanism is \emph{anonymous} if agents with precisely the same valuation functions must have the same probabilities of receiving each item; a mechanism is \emph{nonbossy} if an agent cannot change other agents' allocation without changing its own allocation.
\emph{Incentive-compatibility} is an important property in understanding the role that strategic behavior can play. A mechanism is \emph{manipulable} if an agent can misreport its preferences and improve its utility. A mechanism that is not manipulable is \emph{incentive-compatible} (a.k.a., \emph{truthful}). 

Truthfulness, in a sense, is on top of the properties mentioned above. For example, Pareto efficiency is defined in regard to the declared preferences of the agents, rather than their genuine preferences. Without the truthfulness property, an agent has little or strong incentive to manipulate in a mechanism and could potentially cause significant welfare loss as well as unfairness to the application scenarios.

As one can expect, these properties are not compatible, as shown by several impossibility results. A notable trilemma by \citeauthor{ZHOU:90} (1990) states that there exists no mechanism that is truthful, symmetric, and ex-ante Pareto efficient. \citeauthor{DBLP:journals/jet/BogomolnaiaM01} (2001) shows that there does not exist a truthful, equal treatment of equals, and ordinal efficient mechanism.
Given these impossibility results, one reasonable choice for implementation is the Random Priority mechanism. It is easy to implement and incentive-compatible, but lacks efficiency, as all agents may increase their probabilities of receiving more preferred items by implementing Probabilistic Serial. Moreover, Probabilistic Serial is \emph{SD-envy-free}.
However, it is not truthful. 
In fact, many mechanisms, such as CEEI, is manipulable as well. 
\citeauthor{DBLP:journals/ijgt/EkiciK16} (2016) shows that when agents are not truthful, the outcome of Probabilistic Serial may not satisfy desirable properties related to efficiency and envy-freeness.
Hence, researchers look to develop refined analytic and experimental works to answer a basic question: \emph{which mechanism to employ in practical applications}?

There are at least two scenarios that agents may dispense with manipulation in mechanisms that are not incentive-compatible. One scenario is where computational complexity is considered as an obstacle against manipulation. If manipulation is hard to compute, possibly agents will behave sincerely. The other scenario is by characterizing the extent to which strategic manipulations can increase utilities. If agents can increase their utility by only a small extent, given that there is an inherent cost for them to collect necessary information from other agents in order to compute the best response strategy, perhaps they would behave truthfully.


In this paper, we consider the second scenario in the Probabilistic Serial mechanism. 
We adopt the {\em incentive ratio} notion~\cite{DBLP:conf/esa/ChenDZ11}~\cite{DBLP:conf/icalp/ChenDZZ12} to quantify agents' incentive to deviate from reporting their actual private information. Informally, it is the factor of the largest possible utility gain that an agent achieves by behaving strategically, given that all other agents have their strategies fixed. Our main result is the following theorem.

{\bf Theorem:} In the Probabilistic Serial mechanism, when the number of agents is no less than the number of items, no agent is able to unilaterally manipulate and increase its utility to more than 1.5 times of the utility when reports truthfully.

In understanding the implication of the 3/2 approximation on the manipulation incentives, we put in mind of two technical facts. 
Firstly, the incentive ratio is defined in a \emph{worst-case} sense. Therefore, its bound is the strongest approximation guarantee of the manipulation incentives, in all cases. It also means that the incentive ratio bound ignores the likelihood that extreme cases happen. It could be the case that the probability for an agent to attain the 3/2 times utility is negligible. Indeed, our tight bound example is rather pathological as it is constructed artificially. 
Secondly, our results built upon the \emph{complete information} and \emph{perfect rationality} assumptions. That is, agents have complete information about other agents' preferences and are able to compute the best response strategy accordingly.  
If either of these assumptions is missing, the agents' power to manipulate the mechanism would be much smaller than what the $\frac{3}{2}$ bound implies. 
In fact, computing the best response strategy is intractable in general~\cite{DBLP:conf/atal/AzizGMMNW15}.
Moreover,~\citeauthor{DBLP:journals/geb/Hugh-JonesKV14} (2014) showed that humans do not manipulate Probabilistic Serial mechanism optimally. \citeauthor{HalpernPassSeeman2014} (2014) provided an excellent survey of work using \emph{bounded rationality} in decision theory.
Therefore, a small constant incentive ratio, in particular, 3/2 in our results, indicate that the agents' incentive and ability to manipulate the mechanism is reasonably bounded. 
To support this statement, we further conduct some experiments. The experiments demonstrate that the agents' utility gain by strategic manipulation is rather limited on average, even if there are not many agents and items. 

From the technical point of view, characterizing the incentive ratio of a mechanism is challenging. As a starting point, determining the best response strategy, i.e., a report that maximizes the agent’s expected utility is often intractable. In fact, this task is NP-hard in the context of the Probabilistic Serial mechanism~\cite{DBLP:conf/atal/AzizGMMNW15}, and only local manipulation has been done experimentally by greedy search~\cite{DBLP:conf/ijcai/MennleWPS15}. Our proofs circumvent this obstacle by considering the three possible cases of an agent's utility when it truthful report its preferences. It is relatively easy to bound the incentive ratio in two out of the three cases. Nevertheless, the other case requires a fine-grained analysis in which several in-depth observations are critical.  

\subsection{Related Work}
In the presence of incentives, the random assignment problem has been extensively studied in Computer Science and Economics over the years \cite{ZHOU:90,DG:10,mennle2014axiomatic}.  We refer the interested reader to  surveys \cite{AS:13,SU:11}. 

One of the focal points is how precisely agents might compute beneficial manipulations.
What if it is just too computationally difficult to compute a manipulation~\cite{Bartholdi1989,Bartholdi1991}?
Manipulation has been shown to be computationally hard to compute in many voting situations, e.g.,~\cite{DBLP:conf/aaai/DaviesNW12,DBLP:journals/ai/DaviesKNWX14}. 

For the Probabilistic Serial mechanism,~\citeauthor{DBLP:conf/atal/AzizGMMNW15} (2015) showed that computing the best response (manipulation) under complete information of the other agents’ strategies is NP-hard for Expected Utility maximization, a.k.a., EU-relation, but polynomial-time algorithm exists for the Downward Lexicographic maximization, a.k.a., DL-relation. In addition, they showed that Nash deviations under the Probabilistic Serial mechanism can cycle, but a pure Nash equilibrium is guaranteed to exist. Unfortunately, computing a pure Nash equilibrium is intractable in general.

The empirical work by \citeauthor{DBLP:journals/aamas/HosseiniLC18} (2018) disclosed some  results on the manipulability of the Probabilistic Serial mechanism. Their experiments show that the mechanism is almost always manipulable for various combination of agents and items, and the fraction of strongly manipulable profiles goes to one as the ratio of items to agents increases. However, their results does not inform to what extend these beneficial manipulations are.

An interesting work by~\citeauthor{RePEc:ecm:emetrp:v:78:y:2010:i:5:p:1625-1672} (2010) showed that Random Priority and Probabilistic Serial mechanisms become equivalent when there exist a large number of copies of each item type. 
Their results imply that, on the one hand, the inefficiency of the Random Priority mechanism becomes small and disappears in the limit, as the economy becomes large; on the other hand, the incentive problem of the Probabilistic Serial mechanism disappears in large economies.

The incentive ratio notion was proposed by~\cite{DBLP:conf/esa/ChenDZ11}~\cite{DBLP:conf/icalp/ChenDZZ12}. The authors investigated the buyers' incentive to manipulate Fisher markets. They showed that no agent could gain more than twice and 1.445 times by strategizing in Fisher markets with linear, Leontief utility functions, and Cobb-Douglas utility functions, respectively.



\section{Preliminaries}
The random assignment problem consists of $n$ agents and $m$ divisible items $o_j, j\in[m]$\footnote{In general, $n$ and $m$ are not necessarily equal.  Many works study the case that $n=m$, as in this case, the allocation matrix is doubly stochastic. Thus, any probabilistic allocation can be seen as a convex combination of a set of deterministic allocations, and equivalence between a probabilistic allocation of indivisible items and a fractional assignment of divisible items is established. Our results hold for the case that $n \ge m$, but the proofs cannot be extended to the case $n < m$ straightforwardly.}. 
In the Probabilistic Serial mechanism, agents express strict ordinal preferences, $\succ$, over items. In other words, they are not indifferent between any two items. 
As an \emph{ordinal mechanism}, Probabilistic Serial has several  advantages over cardinal mechanisms that require agents to declare their actual utilities, including simplicity and low communication complexity.
The agents are endowed with von Neumann-Morgenstern utilities over the items. 
We denote the utility derived by agent $i$ on obtaining a unit of item $j$ by $a_{ij}$. Under standard normalization, $0\le a_{ij} \le 1$, $\forall i \in [n], j \in [m]$. Denote vector $\mathbf{a}_i = (a_{i1}, \cdots, a_{1m})$.


The Probabilistic Serial rule collects agents' preference $\succ$ and output an assignment of items to them. Denote an assignment by a matrix $\mathrm{X}=[x_{ij}]_{n \times m}$, where $x_{ij}$ indicates the probability of agent $i$ receiving item $j$, $\sum_{i} x_{ij}=1, \forall j$. So, 
the expected utility of agent $i$ is  $u_i = \sum_{j} a_{ij}x_{ij}$. 
Agents are self-interested and  may misreport their ordinal preferences if that results in a better allocation (from their perspective). In that case, $u_i(\succ_i, \succ_{-i}) < u_i(\succ'_i, \succ_{-i})$, where $\succ_i$ is agent $i$'s true preferences, $\succ_{-i}$ is other agents' preferences, and $\succ'_i$ is a misreport by agent $i$. We adopt the \emph{incentive ratio} notion to characterize the extent to which utilities can be increased by strategic plays of individuals. The {\em incentive ratio} of agent $i$ in mechanism $M$ is\footnote{Here we replace the ordinal relations $\succ$ by cardinal values $a_{ij}$ as the notion also applies to cardinal mechanisms.}
\begin{align*}
    r_i^M=\max_{\mathbf{a}_{-i}} \
\frac{\max_{\mathbf{a}'_i} u'_i(\mathbf{a}'_i, \mathbf{a}_{-i})}{u_i(\mathbf{a}_i, \mathbf{a}_{-i})}.
\end{align*}
Note that the denominator is the utility of agent $i$ when it truthfully reports its preferences, and the numerator is the largest possible utility of agent $i$ when it unilaterally misreports its preferences. The incentive ratio of agent $i$ is then the maximum value of the ratio over all possible inputs of other agents. The incentive ratio of a mechanism $M$ is then $\max_i r_i^M$.  
Throughout the paper, w.l.o.g., we consider the strategic manipulation of agent 1. 

A standard approach to characterize manipulation in mechanism design (for example, voting and market design) is to consider that the manipulating agent has complete knowledge about the other agents' reports. The definition of incentive ratio makes a similar assumption here. Agent 1 knows $\succ_{-1}$ and can figure out its best response strategy $\succ'_1$ accordingly.

The Probabilistic Serial mechanism simulates a \emph{simultaneous eating algorithm}. In the mechanism, agents simultaneously ``eat'' their most preferred items at a uniform speed, moving onto their next most preferred item whenever an item is fully eaten. 
The following example appeared in \cite{DBLP:journals/jet/BogomolnaiaM01} illustrates that the mechanism is not incentive-compatible. 

{\bf Example:} 
There are three agents and three items $a, b,$ and $c$. The agents preferences and the corresponding allocation are

\begin{tabular}{ c c c c c }
  1: $b \succ a \succ c$  &     &   0 & 3/4 & 1/4 \\
  2: $a \succ b \succ c$  & X = & 1/2 & 1/4 & 1/4 \\
  3: $a \succ c \succ b$  &     &   1/2 & 0  & 1/2 
\end{tabular}

If agent 1 misreports its preferences as follows, then the allocation changes accordingly. 

\begin{tabular}{ c c c c c }
  1: $a \succ b \succ c$  &     &   1/3 & 1/2 & 1/6 \\
  2: $a \succ b \succ c$  & X = & 1/3 & 1/2 & 1/6 \\
  3: $a \succ c \succ b$  &     &   1/3 & 0  & 2/3 
\end{tabular}

For some utility function that is compatible with agent 1's true preferences, for example, $a_{1a}=0.9, a_{1b}=1, a_{1c}=0$, its utilities are $u_1=0.75$ in the \emph{truthful profile} and $u'_1=0.8$ in the \emph{manipulation profile}. The high-level intuition behind this manipulative example is that, both items $a$ and $b$ are important to agent 1, but item $a$ is more competitive as the other two agents place it as their most preferred item; so, instead of start eating a less-competitive item $b$, it is better for agent 1 to start with eating item $a$. 


\section{Incentive Ratio Upper Bound}
In this section, we prove the $\frac{3}{2}$ incentive ratio upper bound. 
Our first step is a reduction, which shows that it is sufficient to consider the instances in which agents' utilities are \emph{dichotomous}. That is, agent 1's preferences $a_{1j}$ are either close to 1, or close to 0, $\forall j=1,\cdots,m$.

\begin{lemma}\label{dichotomous}
Given any truthful profile $(\succ_1, \succ_{-1})$ and agent 1's cardinal preference $a_{1j}$'s that are compatible with the ordering $\succ_1$, denote ratio $c=\frac{u'_1}{u_1}$, where $u'_1$ is agent 1's maximum utility attainable by manipulation. Then one can always construct a corresponding dichotomous preference $b_{1j}$'s that are also compatible with $\succ_1$, such that the ratio $c$ is no less than before. 
\end{lemma}

\begin{proof}
W.l.o.g., assume $o_1 \succ_1 o_2 \succ_1 \cdots \succ_1 o_m$ and $1 \ge a_{11} > a_{12} > \cdots > a_{1m} \ge 0$. Denote $l_j$ and $l'_j$ the length of time that agent 1 spent on eating item $o_j$ in the truthful profile and the manipulation profile, respectively. By definition, 
\begin{align*}
    c(a_{1j}) = \frac{u'_1}{u_1} = \frac{a_{11} l'_1 + a_{12} l'_2 + \cdots + a_{1m} l'_m}{a_{11} l_1 + a_{12} l_2 + \cdots + a_{1m} l_m}.
\end{align*}
We will show that by carefully pushing $a_{1j}$'s towards binary values 1 and 0, the ratio $c$ is non-decreasing. Denote 
\begin{align*}
    k = \argmax_j \frac{l'_1 + l'_2 + \cdots + l'_j}{l_1 + l_2 + \cdots + l_j}, \,\,\,\
    c_{\max} = \frac{\sum_{j=1}^{k} l'_j}{\sum_{j=1}^{k} l_j}.
\end{align*}
We can rewrite $c$ as 
\begin{align*}
    c(a_{1j}) = \frac{ \sum_{j=1}^{m} (a_{1j} - a_{1,j+1}) \sum_{h=1}^{j} l'_{h}}{\sum_{j=1}^{m} (a_{1j} - a_{1,j+1}) \sum_{h=1}^{j} l_{h}}, \,\ \textit{where} \,\ a_{1,m+1}=0.
\end{align*}
Now we construct a new preference profile\footnote{In this Lemma, we include the $\epsilon$ terms for the completeness of the proof. In the rest of the paper, we omit these terms when they are needed for expressing a strict ordinal preference, as they do not affect the characterization of the incentive ratio bounds.} $b_{1,j} = 1 - (j-1) \epsilon, j=1,\cdots,k$ and $b_{1,j} = (m - j) \epsilon, j=k+1,\cdots,m$. Note that $b_{1,j}$'s are compatible with $\succ_1$. So, the truthful allocation $l_j$'s remain the same; by using the same strategy, the manipulation allocation $l'_j$'s are also kept the same. 
Moreover, 
\begin{align*}
    c(b_{1,j}) &= \frac{ \sum_{j=1}^{m} (b_{1,j} - b_{1,j+1}) \sum_{h=1}^{j} l'_{h}}{\sum_{j=1}^{m} (b_{1,j} - b_{1,j+1}) \sum_{h=1}^{j} l_{h}} \\
    & = \frac{\epsilon \sum_{j\neq k} \sum_{h=1}^{j} l'_h + (1-(m-2)\epsilon)\sum_{h=1}^{k} l'_h}{\epsilon \sum_{j\neq k} \sum_{h=1}^{j} l_h + (1-(m-2)\epsilon)\sum_{h=1}^{k} l_h} \\
    & \rightarrow c_{\max} \,\,\ (\epsilon \rightarrow 0)
\end{align*}
\end{proof}
The role of this lemma is similar to the \emph{quasi-combinatorial structure} \cite{DBLP:conf/sagt/Filos-RatsikasF014} and the \emph{zero-one principle}~\cite{DBLP:conf/stoc/AzarR04}. It substantiates the intuition that because the mechanism is ordinal, the worst-case incentive ratio is encountered on extreme valuation profiles.

By Lemma~\ref{dichotomous}, we can classify the items into two categories. They are, the set of items that agent 1 is interested in, i.e., $\overline{O} = \{ o_j | a_{1j} \,\ \textit{is close to} \,\ 1 \}$, and the set of items that agent 1 is not interested in, i.e., $\underline{O} = \{ o_j | a_{1j} \,\ \textit{is close to} \,\ 0 \}$. For ease of notation, assume that agent 1 is interested in the first $k$ items. So, $\overline{O} = \{ o_1, \cdots, o_k \}$. In the truthful profile, assume that agent 1 spends time $x_{j}$ on eating item $j$, then its utility is $u_1 \approx \sum_{j=1}^{k} x_j$. Note that agent 1 may not be able to get a positive fraction of each item in $\overline{O}$, so some of these $x_j$'s may be equal to 0.
Agent 1 may also eat some of the items in $\underline{O}$, but their contribution to its utility is negligible.

Given agents' ordinal preferences (no matter they are truthful or not), at any moment $t$, the following lemma compares the amount of each item that is not eaten up yet in two scenarios. In the \emph{normal scenario}, all agents eat items according to their reported ordinal preferences as normal;
in the \emph{pause scenario}, a set of agents is paused from time $t$ for some time while the other agents continue eating normally.  

\begin{lemma}\label{pause}
For any item, at any moment from time $t$ until it is eaten up, the remaining amount of the item in the pause scenario is no less than the remaining amount in the normal scenario.
\end{lemma}
\begin{proof}
We prove it by contradiction. Denote $t_{\inf}$ the earliest moment, after which there exists an item $j^*$, its remaining amount in the pause scenario is less than its remaining amount in the normal scenario. 
For a small enough $t_\delta>0$, at the moment $t_{\inf} + t_{\delta}$, the number of agents who are eating item $j^*$ in the pause scenario must be more than the number of agents who are eating item $j^*$ in the normal scenario. Otherwise, the amount of $j^*$ cannot be smaller in the pause scenario. Denote agent $2$ one of these agents who are eating item $j^*$ in the pause scenario but not in the normal scenario. Recall that every agent eats items in the order from its declared more preferred items to less preferred items. Then agent 2's eating status implies that there exists an item $j'$, which is more preferable to item $j^*$ in agent 2's report, such that at the moment $t_{\inf} + t_{\delta}$, item $j'$ is eaten up in the pause scenario but is not eaten up yet in the normal scenario. However, it contradicts our choice of item $j^*$ and the definition of $t_{\inf}$. 
\end{proof}

This lemma is critical as it is repetitively used in later proofs. It holds no matter how many agents are paused and for how long they are paused. It holds for any inputs to the PS mechanism, regardless the agents are truthful or not. With the presence of Lemma~\ref{pause}, we can better present the high-level ideas of the proofs.

Denote $u_1$ and $u'_1$ agent 1's utility in the truthful profile and the manipulation profile, respectively. Denote $T$ and $T'$ the moment by which all items in $\overline{O}$ are eaten up in the truthful profile and the manipulation profile, respectively. 
Denote $\tilde{T}$ and $\tilde{T}'$ the moment by which all items in $\overline{O}$ are eaten up while agent 1 is paused all the time (or, say, agent 1 is eliminated from the eating process) in the truthful profile and the manipulation profile, respectively.

When we ignore the $\epsilon$ terms for maintaining a strict ordering preference, obviously, $u_1 = T \le T' = u'_1$. By Lemma~\ref{pause}, $T' \le \tilde{T}'$. We also note that $\tilde{T}=\tilde{T}'$, as agent 1's reports have no impact on the eating process when agent 1 is eliminated. Therefore, we can prove our main result $u'_1 \le \frac{3}{2} u_1$ by showing that $\tilde{T} \le \frac{3}{2} T$.
This way, we circumvent the requisition of figuring out agent 1's best response strategies and only need to focus on how much more time it takes other agents to consume the items without the presence of agent 1. This approach is used in proving Theorem~\ref{case1} and Case 1 of Theorem~\ref{case2}.

We prove that our main theorem holds in all of the three possible cases when $m\leq n$, according to the value of $T$. That is, $0 < T < \frac{1}{2}$, $\frac{1}{2} \le T < \frac{2}{3}$, and $\frac{2}{3} \le T \leq 1$.

\begin{theorem}\label{case1}
When $0 < T < \frac{1}{2}$, the Incentive Ratio is upper bounded by $\frac{3}{2}$.
\end{theorem}

\begin{proof}
We will employ Lemma~\ref{pause} and we are interested in a subset of $\overline{O}$. Denote $\overline{O}^* \subseteq \overline{O}$ the set of items that agent 1 gets a positive fraction in the truthful profile. W.l.o.g., assume $\overline{O}^* = \{o_1,\cdots,o_{k^*}\}$. Denote $l_j$ the length of time that agent 1 spent on eating item $o_j, j=1,\cdots,k^*$. That being said, item $o_j$ is eaten up at the moment $\sum_{h=1}^{j}l_h$.
Therefore, we have $T=\sum^{k^*}_{j=1}l_j$.

Since $0 < T < \frac{1}{2}$, and it takes at least time $\frac{1}{2}$ for two agents to eat up an item, we know that at the moment $\sum_{h=1}^{j}l_h$, there are at least three agents eating item $o_j, j=1,\cdots,k^*$. So, apart from agent 1, there are at least two other agents eating each of these items $o_j$ at the moment they are eaten up.

Now let us consider the following process in which agent 1 is eliminated. From the start to $l_1$, all agents eat items according to their preferences. At time $l_1$, pause all other agents but the agents who are eating item $o_1$. Since agent 1 is eliminated, it will take them some extra time $\delta_1$ to eat up item $o_1$. Agent 1 is absent for $l_1$ time so far and there are at least two of these agents, so $\delta_1 \le \frac{l_1}{2}$. At the moment $l_1 + \delta_1$, resume all agents' eating procedure, and pause all agents but the agents who are eating item $o_2$ at the moment $l_1 + \delta_1 + l_2$. For the same reason, it will take these agents an extra $\delta_2 \le \frac{l_2}{2}$ time to eat up the fraction due to agent 1's absence. Repeating this process until all items in $\overline{O}^*$ are eaten up; at this moment all items in $\overline{O} / \overline{O}^*$ are eaten up as well. 
Therefor, we have that 
$\tilde{T} \le \sum_{j=1}^{k^*} (l_j + \delta_j) \le
\sum_{j=1}^{k^*} (l_j + \frac{l_j}{2})
\le \frac{3}{2} T.$
\end{proof}


The second case is the most challenging one. It requires a fine-grained analysis. 

\begin{theorem}\label{case2}
When $\frac{1}{2} \le T < \frac{2}{3}$, the Incentive Ratio is upper bounded by 3/2.
\end{theorem}

In this case, for every item $o_j \in \overline{O}^*$, if at least two other agents are eating the item with agent 1 at the moment $\sum^j_{h=1}l_h$, then the theorem can be proved in the same way as Theorem~\ref{case1}. Therefore, assume that there exists an item that only one other agent is eating it with agent 1 at the moment it is eaten up. Note that there could only exist one such item, as it takes at least time $\frac{1}{2}$ for two agents to eat up one item.
Denote this item by $o_{k'}$ and the other agent by agent 2. For ease of notation, let $t_1=l_1 + \cdots + l_{k'-1}$,  $t_2=l_{k'}$, and $t_3=l_{k'+1} + \cdots + l_{k^*}$. Then $t_1+t_2+t_3=T$.

\begin{lemma}
$t_1 + t_2 \ge \frac{1}{2}$, $t_3 < \frac{1}{6}$, $t_1 < \frac{1}{3}$, and $t_2 > 2 t_3$.
\end{lemma}
\begin{proof}
Since only agents 1 and 2 are eating item $o_{k'}$, it must take at least time $\frac{1}{2}$ for them to eat it up. So $t_1 + t_2 \ge \frac{1}{2}$. In addition, $T < \frac{2}{3}$, so $t_3 = T-(t_1+t_2) < \frac{1}{6}$.
On the one hand, agent 1 has been eating item $o_{k'}$ for time $t_2$, and agent 2, even if it start eating item $o_{k'}$ from the beginning, has been eating it for time $t_1+t_2$, we know that $t_2 + (t_1 + t_2) \ge 1$; on the other hand, $t_1+t_2\le T<\frac{2}{3}$, we conclude that $t_2 > \frac{1}{3}$ and $t_1 < \frac{1}{3}$. Therefore, we have $t_2 > \frac{1}{3} \ge 2t_3$.
\end{proof}

Denote $O_1$ the set of items that are eaten up on or before time $t_1$ in the truthful profile, $O_2$ the set of items that are eaten up in time interval $(t_1, t_1+t_2]$ in the truthful profile, and $O_3$ the set of items that are eaten up in time interval $(t_1+t_2, T]$ in the truthful profile. Note that these three sets contain all items that agent 1 is interested in, i.e., $\overline{O} \subseteq O_1 \cup O_2 \cup O_3$.

Since $t_1 < \frac{1}{3}$, in time interval $[0, t_1]$, there are at least three agents eating any item that agent 1 was eating, so the analysis for this interval is similar to that in Theorem~\ref{case1}.

\begin{corollary}\label{T1}
If we eliminate agent 1, all items in $O_1$ would be eaten up within time $\frac{3}{2} t_1$. 
\end{corollary}

For the set $O_3$, we would not obtain the same $\frac{3}{2}$ bound straightforwardly, but are able to obtain a slightly looser bound, which will be used together with some other approaches for handling $O_2$ to obtain an overall $\frac{3}{2}$ bound. We first show the following lemma.

\begin{lemma}\label{T3twoagents}
In the normal scenario, for each item in $O_3$, at the moment that it is finished, there are at least two agents other than agents 1 and 2 who are eating the item.
\end{lemma}
\begin{proof}
For each item $o_j \in O_3$, we prove the lemma in three possible cases.
\begin{itemize}
\item At the moment $t_1 + t_2$, no agent is eating item $o_j$. In this case, because $t_3 < \frac{1}{6}$ and the item $o_j$ is eaten up within the interval $[t_1+t_2, t_1+t_2+t_3]$, there must be at least six agents eating the item. Amongst these six agents, even if two of them are agents 1 and 2, there are another four agents. 
\item At the moment $t_1 + t_2$, only one agent is eating item $o_j$. In this case, even if this agent eats item $o_j$ from the start, there are at least $1-(t_1+t_2)$ fraction of this item remaining, and it will be eaten up before the moment $t_1+t_2+t_3$, by $\frac{1-(t_1+t_2)}{t_3} = 1 + \frac{1-t}{t_3} > 1+ \frac{1-2/3}{1/6} = 3$, we know that there are at least four agents eating the item. Two of them might be agents 1 and 2. Even though, there are at least another two agents. 
\item At the moment $t_1 + t_2$, at least two agents are eating item $o_j$. Since agents 1 and 2 are eating item $o_{k'}$,  these must be two other agents. 
\end{itemize}
\end{proof}

If two agents were absent for some time, it will take another two agents the same amount of time to eat up the amount of items left over due to their absence. 
Therefore, a direct consequence of the above lemma is the following bound.

\begin{corollary}\label{T3}
After eliminating agent 1 from eating items in $O_1$, if we eliminate agents 1 and 2 from the moment $\frac{3}{2} t_1 + t_2$, all items in $O_3$ will be eaten up by at most an extra $t_3$ time.
\end{corollary}

We now turn to prove Theorem~\ref{case2} by combining these intermediate results and an analysis on item $o_{k'}$ and the set $O_3$. According to the high-level idea of our proofs described before, we will eliminate agent 1 from eating item $o_{k'}$. This will delay the moment that item $o_{k'}$ is eaten up. More importantly, it will lead to two possible consequences. 
 \\

\noindent\textit{Proof of} Theorem~\ref{case2}: \\
{\bf Case 1:} The process of eating item $o_{k'}$ is extended, so some agents who eat items in $O_3$ will not continue eating their next item in $O_3$ and will take the chance to eat item $o_{k'}$ before the moment $\frac{3}{2} t_1 + t_2 + 2t_3$. 

In the truthful profile, denote $\{s_1,\cdots,s_j\} \subseteq O_3$ the set of items that either agent 1 or agent 2 gets a positive fraction.   
For $h=1,\cdots, j$, denote $c_h$ and $d_h$ the fraction of item $s_h$ that agents 1 and 2 get, respectively.

In the manipulation profile (agent 1 is eliminated, agent 2 is eating item $o_{k'}$), denote $z_h$ the moment at which item $s_h$ is eaten up; denote agent 3 the first agent who comes to eat item $o_{k'}$ before the moment $\frac{3}{2} t_1 + t_2 + 2t_3$. In case there are multiple agents come to agent item $o_{k'}$ at the same time, pick one of them as agent 3 randomly.
Denote $\frac{3}{2} t_1 + t_2 + x$ the moment agent 3 starts to eat item $o_{k'}$, where $0 \le x \le 2t_3$.

Assume $1.5t_1 + t_2 +x \in (z_{w-1},z_w], w\le j$.
At the moment $\frac{3}{2} t_1 + t_2 + x$, compare the truthful profile and the manipulation profile, due to the absence of agent 1, item $o_{k'}$ is not eaten up in time. So both agents 1 and 2 have not started eating items in $O_3$ at time $\frac{3}{2}t_1 + t_2$. So there is a delay of $\frac{1}{2} t_1 + \frac{\sum_{j=1}^{w-1} c_j + \sum_{j=1}^{w-1} d_j}{2}$. The $\frac{1}{2}$ in the second term is due to Lemma~\ref{T3twoagents}.

Now let agents 2 and 3 eat item $o_{k'}$ and pause all other agents until the item is eaten up. This will take agents 2 and 3 time $\frac{t_2 - x}{2}$, here $t_2$ is due to the absence of agent 1, $x$ is how much agent 2 has eaten, the denominator 2 is due to Lemma \ref{T3twoagents}.

Now at time $\frac{3t_1}{2} + t_2 + x+ \frac{t_2-x}{2}$, resume all agents except agent 1.  
We calculate the amount of time after this moment when the items in $O_3$ will be eaten up in the manipulation profile.
It consists of three categories. 
In the manipulation profile, when we pause the all other agents but the agents who are eating item in $\overline{O} \cap O_3$, this moment belongs to the first category.
The time in the first category is at most
$\frac{\sum_{h=w}^j d_h}{2}$. Note that $t_3=\sum_{h=1}^j d_h$, the time in the first category is at most $$\frac{t_3-\sum_{h=1}^{w-1} d_h}{2}.$$

The second category is associated with agent 2. At time $\frac{3t_1}{2} + t_2 + x+\frac{t_2-x}{2}$, agent 2 will begin to eat an item denoted by $o_q$. Then in the truthful profile, at the moment $t_1+t_2+x-\frac{\sum_{j=1}^{w-1} c_j + \sum_{j=1}^{w-1} d_j}{2}$ agent 2 is eating item $o_q$.  
If $o_q\notin O_3$, it implies $o_q$ will not be eaten up at moment $t_1+t_2+t_3$ in the truthful profile. Agent 2 will not affect the moment items in $O_3$ been eaten up in the manipulation profile anymore.
The second category would be empty.
If $o_q\in O_3$, in the manipulation profile, we would pause all other agents but the agents who are eating item $o_q$.
This pause moment belongs to the second category. 
In the manipulation profile, before the moment $1.5t_1+t_2+x$, the total time agent 2 has been absent for eating item in $O_3$  
is $x-\frac{\sum^{w-1}_{h=1}{(c_h+d_h)}}{2}$. These items are $o_1,...,o_{w_1}$ and $o_q$. The length of time agent 2 has been absent for eating item $o_q$ is
$x-\frac{3\sum^{w-1}_{h=1}c_h+\sum^{w-1}_{h=1}d_h}{2}$. In the pause period, it will take agents who are eating $o_q$ at most $$\frac{x}{2}-\frac{3\sum^{w-1}_{h=1}c_h+\sum^{w-1}_{h=1}d_h}{4}$$ time, which is the upper bound of the time in the second category.

The other moments before $O_3$ have been eaten up belongs to third category. In this category, manipulation profile is identical to the truthful profile except agent 1 is not eating. The time in this category would be $$t_3-(x-\frac{\sum_{j=1}^{w-1} c_j + \sum_{j=1}^{w-1} d_j}{2}).$$

Considering all three categories, we upper bound the time when all items in $O_1 \cup O_2 \cup O_3$ are eaten up as follows
\begin{eqnarray*}
\tilde{T}&\leq& [\frac{3t_1}{2} + t_2 +x+ \frac{t_2-x}{2}]
    +[\frac{t_3-\sum_{h=1}^{w-1} d_h}{2}]\\
    &&
    +[\frac{x}{2}-\frac{3\sum_{h=1}^{w-1}c_h}{4}-\frac{\sum_{h=1}^{w-1}d_h}{4}]\\
    &&+[t_3-x+ \frac{c_1+...+c_{w-1}+d_1+...+d_{w-1}}{2}]\\
    &=&1.5t_1+1.5t_2+1.5t_3-
    \frac{\sum_{h=1}^{w-1}(c_h+d_h)}{4}\\
    &\leq& 1.5(t_1+t_2+t_3)
\end{eqnarray*}

{\bf Case 2:} 
Before time $1.5t_1+t_2+2t_3$, the agents who eat items in $O_3$ will continue eating their next items in $O_3$ as they are more favorable than $o_{k'}$.
In this case, we are not going to use the high-level idea presented before to show $\tilde{T} \le \frac{3}{2} T$. Instead, we will characterize an upper bound of agent 1's utility when it uses the best response strategy. This upper bound will be partitioned into two quantities; one of which will be bounded using Lemma~\ref{pause}. 

By using the best response strategy, 
\begin{align*}
    u'_1 = u'_1|_{\le \frac{3}{2} t_1 + t_2 + 2t_3} + u'_1|_{> \frac{3}{2} t_1 + t_2 + 2t_3},
\end{align*}
where $u'_1|_{\le \frac{3}{2} t_1 + t_2 + 2t_3}$ and $u'_1|_{> \frac{3}{2} t_1 + t_2 + 2t_3}$ denote agent 1's utility gained before and after the moment $\frac{3}{2} t_1 + t_2 + 2t_3$, respectively. Obviously, $u'_1|_{\le \frac{3}{2} t_1 + t_2 + 2t_3} \le \frac{3}{2} t_1 + t_2 + 2t_3$, where the equality holds if agent 1 eats items in $\overline{O}$ from the start to time $\frac{3}{2} t_1 + t_2 + 2t_3$. Next, we upper bound $u'_1|_{> \frac{3}{2} t_1 + t_2 + 2t_3}$.

In fact, the moment $\frac{3}{2} t_1 + t_2 + 2t_3$ is calibrated as it is the time needed to eat up items in $O_1$ and $O_3$, plus agent 2 has been eating item $o_{k'}$ alone for $t_2$ time, when agent 1 is eliminated. Now let us check out the items that would have been left over in the manipulation profile at the moment $\frac{3}{2} t_1 + t_2 + 2t_3$. 
By Lemma~\ref{pause}, we notice that the only item in $\overline{O}$ would possibly be still not eaten up yet, is $o_{k'}$. Due to the absence of agent 1, it has $t_2$ left, but agent 2 has been eaten it for an additional $2t_3$ time, so there would be $t_2 - 2t_3$ of item $o_{k'}$ left. 

If $o_{k'}$ is not eaten up yet, agent 2 must be eating it at the moment $1.5t_1+t_2+2t_3$. To gain more utility, the optimal strategy for agent 1 is to eat item $o_{k'}$ after this moment. Agent 1 can get at most $\frac{t_2 - 2t_3}{2}$ fraction for the existence of agent 2. Therefore, 
\begin{align*}
    u'_1 &= u'_1|_{\le \frac{3}{2} t_1 + t_2 + 2t_3} + u'_1|_{> \frac{3}{2} t_1 + t_2 + 2t_3} \\
    &\le \frac{3}{2} t_1 + t_2 + 2t_3 + \frac{t_2 - 2t_3}{2} \\
    & = \frac{3}{2} t_1 + \frac{3}{2} t_2 + t_3 \le \frac{3}{2} T.
\end{align*}
\qed


At last, we have the third case.

\begin{theorem}\label{case3}
When $\frac{2}{3} \le T \le 1$, the Incentive Ratio is upper bounded by $\frac{3}{2}$.
\end{theorem}
This is trivial, since the optimal utility $u_1'$ is upper bound by $\frac{m}{n}\leq 1$. Hence the utility achieved in truthful profile is quite large compared to $u_1'$. Formally, $r^{PS} = \max \frac{u'_1}{u_1} \le \frac{1}{\frac{2}{3}} \le \frac{3}{2}$.

Combining Theorems~\ref{case1},~\ref{case2},~\ref{case3}, we complete the proof of our main Theorem.


\section{The Tight Bound Example}
In this section, we show that the upper bound $\frac{3}{2}$ is tight. 
The following instance and the manipulation provide us a tight lower bound example of the incentive ratio.  

In the instance, the $n$ agents' preferences are 

\begin{tabular}{ c c c c c }
  &1:&   $o_1 \succ o_2 \succ \cdots \succ o_{n-1} \succ o_n$  &  \\
  &2:&   $o_1 \succ o_2 \succ \cdots \succ o_{n-1} \succ o_n$  &  \\
 &$i$:&  $o_2 \succ o_3 \succ \cdots \succ o_{n} \,\,\,\,\ \succ o_1$ &    $i=3,\cdots,n$
\end{tabular}

Agent 1 is interested in the first $\frac{n}{2} - 1$ items, i.e., $\overline{O} = \{o_1, \cdots, o_{\frac{n}{2} - 1}\}$. Then in the truthful profile, $u_1=\frac{1}{2}$.
Agent 1 will only get half fraction of item 1.

By using the strategy $o_2 \succ o_3 \succ \cdots \succ o_{\frac{n}{2}-1} \succ o_1 \succ o_{\frac{n}{2}} \succ \cdots \succ o_n$. Agent 1 will get $\frac{1}{n-1}$ fraction of item 2 to $\frac{n}{2}-1$ and $\frac{1}{4}$ fraction of item 1. Agent 1's utility becomes $u'_1=\frac{3}{4}$. So the ratio is $\frac{3}{2}$.


\section{Experimental Evaluation}
In this section, we present numerical experiments on the extent to which an agent can increase its utility by unilateral manipulation in the Probabilistic Serial mechanism. While we investigated the theoretical incentive bounds of the mechanism in the worst-case framework, the purpose of this section is to evaluate its performance in an average-case framework. 

We set up our experiments as follows. We set $n=m$, i.e.,  the number of agents is equal to the number of items. We vary this number from 8 to 20. For each value of $n$, we generate 10000 of these instances. We construct an instance by uniformly at random and independently generating each agent's ordinal preferences. We make the manipulator's cardinal preferences dichotomous. This is because of Lemma~\ref{dichotomous} -- it always generates a larger utility gain comparing to the corresponding non-dichotomous preferences. We vary the number of items the manipulator is interested in, say $k$, from 2 to 6. 
For each of these instances, we enumerate the manipulator's all $k!$ strategies, in order to figure out the largest possible utility the agent can obtain. By dividing the largest utility attainable of its utility obtained in the truthful profile, we get a ratio to evaluate the agent's utility increment.   

Our experimental results are presented in Figure 1. We can observe that the ratio is between 1.02 and 1.06, which is much smaller than the worst-case incentive ratio bound of 1.5. In particular, for instances with a fixed number of agents/items, the expected ratio is decreasing in the number of items that the manipulator is interested. For instances with a fixed number of items that the manipulator is interested, the expected ratio is increasing in the number of agents/items. 

\begin{figure}[h]
    \centering
    \includegraphics[width=0.7\columnwidth]{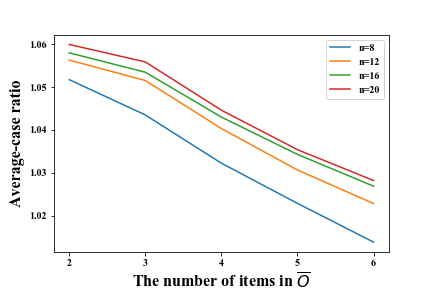}
    \caption{Experimental Evaluation}
    \label{fig:my_label}
\end{figure}


\section{Conclusion}
By knowing that the Probabilistic Serial mechanism is not incentive-compatible, in this paper, we examined the degree of an agent's incentive to manipulate the mechanism. In the form of incentive ratio, we showed that no agent is able to increase its utility by a 50\% through strategic behaviors. This ratio is the strongest guarantee in the worst-case sense, by allowing that the agent has complete information about other agents' private information and is perfectly rational. In addition, we conduct experiments to examine the manipulation incentive in an average-case sense. The evaluation demonstrated that the utility-incremental ratio is much smaller than the theoretical incentive ratio bound. 
The evaluation bears out the supposition that Probabilistic Serial is approximately incentive-compatible in practice, even for small size instances. 
We hope that this work offered a better understanding of the robustness of the PS mechanisms against manipulations, which is one step beyond knowing that the mechanism could be manipulated. Perhaps the next step could be a formal analytical characterization of the average-case incentive ratio.

\section{ Acknowledgments}
Zihe Wang was supported by the Shanghai Sailing Program (Grant No. 18YF1407900) and the National NSFC (Grant No. 61806121). Part of this work was done when Jie Zhang was visiting Peking University.

\bibliographystyle{aaai}
\bibliography{IR}

\end{document}